\documentclass[paper=a4, fontsize=11pt]{scrartcl}
\usepackage[T1]{fontenc}
\usepackage{fourier}
\usepackage[english]{babel}								 
\usepackage{amssymb}
\usepackage{amsmath,amsfonts,amsthm} 
\usepackage[pdftex]{graphicx}	
\usepackage{url}
\usepackage[title,titletoc,page]{appendix} 
\usepackage{listings} 
\usepackage{color}
\usepackage{cite}
\usepackage{caption}
\usepackage{subcaption}
\usepackage[ruled,vlined]{algorithm2e}
\usepackage{float}
\usepackage{booktabs}

\numberwithin{equation}{section}		
\numberwithin{figure}{section}			
\numberwithin{table}{section}				

\newtheorem{defi}{Definition}[section]
\newtheorem{lem}{Lemma}[section]
\newtheorem{thm}{Theorem}[section]

\newtheorem{prop}{Proposition}[section]

\newcommand{\bnorm}[1]{\left|\left|{#1}\right|\right|}

\newcommand{\vect}[1]{\boldsymbol{\mathbf{#1}}}

\title{Super-resolution of positive near-colliding point sources
	\thanks{\footnotesize This work was supported in part by the Swiss National Science Foundation grant number
		200021--200307.}}
\author{Ping Liu\thanks{\footnotesize Department of Mathematics, ETH Z\"urich, R\"amistrasse 101, CH-8092 Z\"urich, Switzerland (ping.liu@sam.math.ethz.ch,  habib.ammari@math.ethz.ch).}  \and Habib Ammari\footnotemark[2]}

\date{}
\begin{document}
\maketitle
\begin{abstract}
In this paper, we analyze the capacity of super-resolution of one-dimensional positive sources.  In particular, we consider the same setting as in \cite{batenkov2019super} and generalize the results there to the case of super-resolving positive sources. To be more specific, we consider resolving $d$ positive point sources with $p \leqslant d$ nodes closely-spaced and forming a cluster, while the rest of the nodes are well separated. Similarly to \cite{batenkov2019super}, our results show that when the noise level $\epsilon \lesssim \mathrm{SRF}^{-2 p+1}$, where  $\mathrm{SRF}=(\Omega \Delta)^{-1}$ with $\Omega$ being the cutoff frequency and $\Delta$ the minimal separation between the nodes, the minimax error rate for reconstructing  the cluster nodes is of order $\frac{1}{\Omega} \mathrm{SRF}^{2 p-2} \epsilon$, while for recovering the corresponding amplitudes $\left\{a_j\right\}$ the rate is of  order $\mathrm{SRF}^{2 p-1} \epsilon$. For the non-cluster nodes, the corresponding minimax rates for the recovery of nodes and amplitudes are of order $\frac{\epsilon}{\Omega}$ and $\epsilon$, respectively. Our numerical experiments show that the Matrix Pencil method achieves the above optimal bounds when resolving the positive sources.
\end{abstract}

\section{Introduction}

In recent years, the problem of super-resolution (SR), which seeks to extract fine details of a signal from its noisy Fourier data in a bounded frequency domain, draws increasing interest in the field of applied mathematics. In particular, considerable progress has been made in the study of super-resolution of sparse signals, e.g.  a host of algorithms  \cite{candes2014towards, duval2015exact, poon2019, tang2013compressed, tang2014near, morgenshtern2016super, morgenshtern2020super, denoyelle2017support} were devised for resolving signals with a sparse prior. The sparse signals are frequently modeled as discrete measures
$$
F(x)=\sum_{j=1}^d a_j \delta\left(x-x_j\right), \quad x_j \in \mathbb{R},
$$
where $\delta$ is the Dirac's $\delta$-distribution. Let $\mathcal{F}[F]$ denote the Fourier transform of $F$ :
$$
\mathcal{F}[F](s)=\int_{-\infty}^{\infty} F(x) e^{-2 \pi i s x} \mathrm{~d} x .
$$
The noisy spectral data of the signal $F$ is modeled as a function $\Phi$ satisfying, 
\begin{equation}\label{equ:modelsetting1}
|\Phi(s)-\mathcal{F}[F](s)| \leqslant \epsilon, \quad s \in[-\Omega, \Omega],
\end{equation}
where $\epsilon>0$ represents the noise level and $\Omega>0$ is the cutoff frequency. The sparse SR problem considered in this paper reads: given $\Phi$ as above, estimate the unknown parameters of $F$, namely the amplitudes $\left\{a_j\right\}$ and the nodes $\left\{x_j\right\}$. The minimax error rate for recovering the nodes and amplitudes from the spectral data $\Phi(s)$ has been established in \cite{batenkov2019super}. In the present paper, we aim at exploring the corresponding minimax error rate for resolving positive signals. Since our result is a generalization of the result in \cite{batenkov2019super} which deals with complex signals, we utilize  the same notation, concepts, and configurations as those in \cite{batenkov2019super} for the sake of consistency of the two papers and the convenience of reading.

\subsection{Main contribution}
The main contribution of this paper is the generalization of the estimates in \cite{batenkov2019super} for the minimax error rate for complex signals in the off-the-grid setting  to the case of resolving positive signals.  We consider the case where the nodes $\left\{x_j\right\}$ can take arbitrary real values and the amplitudes $\left\{a_j\right\}$ are known to be positive. We consider the same distribution of nodes as in \cite{batenkov2019super} where it is assumed that $p$ nodes (approximately uniformly distributed), $x_\kappa, \ldots, x_{\kappa+p-1}$, form a small cluster and the rest of the nodes are away from all the other nodes (see Definition \ref{defi:clusterconfig} below). We show in Theorem \ref{thm:minimax1} that for $\epsilon \lesssim(\Omega \Delta)^{2p-1}$ with $\Delta$ being the minimum separation of the clustered nodes, in the worst-case scenario, the errors of recovered nodes $x_j'$ and amplitudes $a_j'>0$ by any minimax algorithm (see Definition \ref{defi:minimaxerror} below), satisfy  
\begin{itemize}
\item For the  non-cluster nodes:
\begin{align*}
	&\max _{j \notin\{\kappa, \ldots, \kappa+p-1\}}\left|x_j-x_j^{\prime}\right| \asymp \frac{\epsilon}{\Omega}, \\
	&\max _{j \notin\{\kappa, \ldots, \kappa+p-1\}}\left|a_j-a_j^{\prime}\right| \asymp \epsilon ;
\end{align*}
\item For the cluster nodes:
\begin{align*}
	&\max _{j \in\{\kappa, \ldots, \kappa+p-1\}}\left|x_j-x_j^{\prime}\right| \asymp \frac{\epsilon}{\Omega}(\Omega \Delta)^{-2 p+2}, \\
	&\max _{j \in\{\kappa, \ldots, \kappa+p-1\}}\left|a_j-a_j^{\prime}\right| \asymp \epsilon(\Omega \Delta)^{-2 p+1}.
\end{align*}
\end{itemize}

Our results reveal that the minimax error rates for recovering the nodes and amplitudes of positive signals in the super-resolution problem are the same as those for resolving general complex signals \cite{batenkov2019super}. To be more specific, for $\epsilon \lesssim \mathrm{SRF}^{-2p+1}$ with $\mathrm{SRF}:= (\Omega \Delta)^{-1}$, the minimax error rate for reconstructing the cluster nodes of positive signals is of the order $(\mathrm{SRF})^{2p-2} \frac{\epsilon}{\Omega}$, while for recovering the corresponding amplitudes the rate is of the order $(\mathrm{SRF})^{2p-1}\epsilon$. On the other hand, the corresponding minimax rates for the recovery of the non-cluster nodes and amplitudes are of the order $\frac{\epsilon}{\Omega}$ and $\epsilon$, respectively. These also indicate that the non-cluster nodes $\left\{x_j\right\}_{j \notin\{\kappa, \ldots, \kappa+p-1\}}$ can be recovered with much better stability than the cluster nodes. 

The main novelty we rely in analyzing the case of positive signals lies in a crucial observation in estimating the lower bound of diameter of the error set (Definition \ref{defi:errorset}). In particular, in Theorem \ref{thm:lowerexampleconstru1}, we observe and demonstrate that the recovered and the underlying signals in the example constructed in \cite{batenkov2019super} can actually be positive signals at the same time. 

We also examine the performance limit of the Matrix Pencil method in resolving positive signals by numerical experiments. The oberved error amplification in the experiments exactly verifies our theory for the minimax error rate. This also indicates that the Matrix Pencil method has the optimal performance in super-resolving positive sources.  

\subsection{Related work and discussion}
In 1992, Donoho first studied the possibility and difficulties of super-resolving multiple sources from  noisy measurements.  In particular, he considered measures supported on a lattice $\{k\Delta\}_{k=-\infty}^{\infty}$ and regularized by a so-called “Rayleigh index”. The measurement is then the noisy Fourier transform of the discrete measure with cutoff frequency $\Omega$.  He derived both the lower and upper bounds for the minimax error of the amplitude recovery in terms of the noise level, grid spacing, cutoff frequency, and Rayleigh index. His results emphasize the importance of the sparsity in the super-resolution. The results were improved in recent years for the case when resolving $n$-sparse on-the-grid sources \cite{demanet2015recoverability}. Concretely, the authors of \cite{demanet2015recoverability} showed that the minimax error rate for amplitudes recovery scales like $\mathrm{SRF}^{2n-1}\epsilon$, where $\epsilon$ is the noise level and $\mathrm{SRF}:= \frac{1}{\Delta\Omega}$ is the super-resolution factor. Similar results for multi-clumps cases were also derived in \cite{batenkov2020conditioning, li2021stable}.

A closely related work to the present paper is \cite{batenkov2019super}, in which the authors derived sharp minimax errors for the location and the amplitude recovery of off-the-grid sources. They showed that for complex sources satisfying the $(p,h,T,\tau, \eta)$-clustered configuration (Definition \ref{defi:clusterconfig}) and $\epsilon \lesssim \mathrm{SRF}^{-2p+1}$ with $p$ being the number of the cluster nodes, the minimax error rate for reconstructing of the cluster nodes is of the order $(\mathrm{SRF})^{2p-2} \frac{\epsilon}{\Omega}$, while for recovering the corresponding amplitudes the rate is of the order $(\mathrm{SRF})^{2p-1}\epsilon$. Moreover, the corresponding minimax rates for the recovery of the non-cluster nodes and amplitudes are of the order $\frac{\epsilon}{\Omega}$ and $\epsilon$, respectively. As mentioned above, in the present paper we have generalized these results to the case when resolving positive sources. Thus, the minimax error estimations for super-resolving both one-dimensional complex and positive sources are well established now.

On the other hand, in order to characterize the exact resolution in the number and location recovery, in the earlier works \cite{liu2021mathematicaloned, liu2021theorylse, liu2021mathematicalhighd, liu2022nearly, liu2022mathematicalpositive, liu2022mathematicalSR} the authors have defined the so-called "computational resolution limits", which characterize the minimum required distance between point sources so that their number and locations can be stably resolved under certain noise level. It was shown that the computational resolution limits for the number and location recoveries in the $k$-dimensional super-resolution problem should be bounded above by respectively $\frac{C_{num}(k, n)}{\Omega}\left(\frac{\sigma}{m_{\min }}\right)^{\frac{1}{2 n-2}}$ and $\frac{C_{supp}(k, n)}{\Omega}\left(\frac{\sigma}{m_{\min }}\right)^{\frac{1}{2 n-1}}$, where $C_{num}(k,n)$ and $C_{supp}(k,n)$ are certain constants depending only on the source number $n$ and the space dimensionality $k$. In particular, these results were generalized to the case when resolving positive sources in \cite{liu2022mathematicalpositive}. In this paper, a similar idea is used to generalize the miminax error estimate to the positive cases.  

 For other works related to the limit of super-resolution, we refer the readers to \cite{moitra2015super,chen2020algorithmic} for understanding the resolution limit from the perceptive of sample complexity and to  \cite{tang2015resolution, da2020stable} for the resolving limit of some algorithms.
 
 For the super-resolution of positive sources, to the best of our knowledge, the theoretical possibility for the super-resolution of positive sources was first considered in \cite{donoho1992maximum}, where the authors characterized the relation between the sparsity of the on-the-grid signal $ \mathbf x$ and the possibility of super-resolution in certain sense. Their definition and results focused on the possibility of overcoming Rayleigh limit in the presence of sufficiently small noise, while our work analyzes the non-asymptotic behavior of the reconstructions. 
 
 In recent years, some researchers analyzed the stability of some super-resolution algorithms in a non-asymptotic regime \cite{morgenshtern2016super, morgenshtern2020super, denoyelle2017support} and derived similar stability results to those proved in this paper, which exhibit the optimal performance of these algorithms.

\subsection{Organization of the paper}
The paper is organized in the following way. Section 2 presents the main results of the minimax error rate and Section 3 exhibits the performance of Matrix Pencil method by numerical experiments. Section 4 proves the main results stated in Section 2.

\section{Minimax bound for the location and amplitude recoveries}\label{section:minimaxerror}
In this section, we present minimax error estimates for the location and amplitude recoveries in the super-resolution of positive signals.

\subsection{Notation and preliminaries}
We shall denote by $\mathcal{P}_d, \mathcal{P}_d^+$ the parameter space of respectively general and positive signals $F$ with amplitudes $a_j$'s and real, distinct and ordered nodes $x_j$'s:
\begin{align*}
&\mathcal{P}_d=\left\{(\mathbf{a}, \mathbf{x}): \mathbf{a}=\left(a_1, \ldots, a_d\right) \in \mathbb{C}^d, \mathbf{x}=\left(x_1, \ldots, x_d\right) \in \mathbb{R}^d, x_1<x_2<\ldots<x_d\right\},\\
&\mathcal{P}_d^+=\left\{(\mathbf{a}, \mathbf{x}): \mathbf{a}=\left(a_1, \ldots, a_d\right) \in (\mathbb{R}^+)^d, \mathbf{x}=\left(x_1, \ldots, x_d\right) \in \mathbb{R}^d, x_1<x_2<\ldots<x_d\right\},
\end{align*}
and identify the $F$'s with their parameters $(\mathbf{a}, \mathbf{x}) \in \mathcal{P}_d$ or $\mathcal{P}_d^+$. We denote
\[
\|F\|_{\infty} =\max \left(\|\mathbf{a}\|_{\infty},\|\mathbf{x}\|_{\infty}\right) .
\]
We shall also denote the orthogonal coordinate projections of a signal $F$ to the $j$-th node and $j$-th amplitude, respectively, by $P_{\mathbf{x}, j}: \mathcal{P}_d ( \mathcal{P}_d^+) \rightarrow \mathbb{R}$ and $P_{\mathbf{a}, j}: \mathcal{P}_d ( \mathcal{P}_d^+) \rightarrow \mathbb{C}(\mathbb R^+)$. 

Let $L_{\infty}[-\Omega, \Omega]$ be the space of bounded complex-valued functions defined on $[-\Omega, \Omega]$ with the norm $\|e\|_{\infty}=\max _{|s| \leqslant \Omega}|e(s)|$.

\begin{defi}
Given $\Omega>0$ and $U \subseteq \mathcal{P}_d^+$, we denote by $\mathfrak{F}(\Omega, U)$ the class of all admissible reconstruction algorithms, i.e.,
$$
\mathfrak{F}(\Omega, U)=\left\{f: L_{\infty}[-\Omega, \Omega] \rightarrow U\right\}. 
$$
\end{defi}

\begin{defi}\label{defi:minimaxerror}
Let $U \subset \mathcal{P}_d^+$. We consider the minimax error rate in estimating a signal $F \in U$ from $\Omega$-bandlimited data as in (\ref{equ:modelsetting1}), with a measurement error $\epsilon>0$:
\[
\mathcal{E}^+(\epsilon, U, \Omega)=\inf _{\mathfrak{f} \in \mathfrak{F}(\Omega, U)} \sup _{F \in U} \sup _{\|e\|_{\infty} \leqslant \epsilon}\|F-\mathfrak{f}(\mathcal{F}[F]+e)\|_{\infty}.
\]
\end{defi}

Note that in order to analyze how the minimax error rate relates to the separation of the nodes and the magnitudes of the amplitudes, we will consider $U\in \mathcal P_{d}^+$ with certain specific constraints in the following discussions.

Similarly, the minimax errors of estimating the individual nodes and the amplitudes of $F \in U$ are defined respectively by 
\[
\begin{aligned}
&\mathcal{E}^{+,\mathbf{x}, j}(\epsilon, U, \Omega)=\inf _{\mathfrak{f} \in \mathfrak{F}(\Omega, U)} \sup _{F \in U} \sup _{\|e\| \leqslant \epsilon}\left|P_{\mathbf{x}, j}(F)-P_{\mathbf{x}, j}(\mathfrak{f}(\mathcal{F}(F)+e))\right|, \\
&\mathcal{E}^{+,\mathbf{a}, j}(\epsilon, U, \Omega)=\inf _{\mathfrak{f} \in \mathfrak{F}(\Omega, U)} \sup _{F \in U} \sup _{\|e\|_{\infty} \leqslant \epsilon}\left|P_{\mathbf{a}, j}(F)-P_{\mathbf{a}, j}(\mathfrak{f}(\mathcal{F}(F)+e))\right| .
\end{aligned}
\]

For a fixed $F \in \mathcal{P}_d^+$, we define the positive and general $\epsilon$-error set as follows. 

\begin{defi}\label{defi:errorset}
The error set of positive signals $E_{\epsilon, \Omega}^{+}(F) \subset \mathcal{P}_d^{+}$ is the set consisting of all the signals $\hat F \in \mathcal{P}_d^{+}$ with
\begin{equation}\label{equ:errorset1}
\left|\mathcal{F}[\hat F](\omega)-\mathcal{F}[F](\omega)\right| \leqslant \epsilon, \quad \omega \in[-\Omega, \Omega] .
\end{equation}
Moreover, the error set of general signal $E_{\epsilon, \Omega}(F) \subset \mathcal{P}_d$ is the set consisting of all the $\hat F \in \mathcal{P}_d$ satisfying (\ref{equ:errorset1}).
\end{defi}

We will denote by $E_\epsilon^{+, \mathbf{x}, j}(F)=E_{\epsilon, \Omega}^{+, \mathbf{x}, j}(F)$ and $E_\epsilon^{+, \mathbf{a}, j}(F)=E_{\epsilon, \Omega}^{+, \mathbf{a}, j}(F)$ the projections of the error set of positive signals onto the individual nodes and the amplitudes components, respectively:
$$
\begin{aligned}
&E_{\epsilon, \Omega}^{+, \mathbf{x}, j}(F)=\left\{\mathbf{x}_j^{\prime} \in \mathbb{R}:\left(\mathbf{a}^{\prime}, \mathbf{x}^{\prime}\right) \in E_{\epsilon, \Omega}^{+}(F)\right\} \equiv P_{\mathbf{x}, j} E_{\epsilon, \Omega}^{+}(F), \\
&E_{\epsilon, \Omega}^{+, \mathbf{a}, j}(F)=\left\{\mathbf{a}_j^{\prime} \in \mathbb{R}^+:\left(\mathbf{a}^{\prime}, \mathbf{x}^{\prime}\right) \in E_{\epsilon, \Omega}^{+}(F)\right\} \equiv P_{\mathbf{a}_{,} j} E_{\epsilon, \Omega}^+(F).
\end{aligned}
$$

Furthermore, we denote by $E_\epsilon^{\mathbf{x}, j}(F)=E_{\epsilon, \Omega}^{\mathbf{x}, j}(F)$ and $E_\epsilon^{\mathbf{a}, j}(F)=E_{\epsilon, \Omega}^{\mathbf{a}, j}(F)$ the projections of the error set of general signals onto the individual nodes and the amplitudes components, respectively:
$$
\begin{aligned}
	&E_{\epsilon, \Omega}^{\mathbf{x}, j}(F)=\left\{\mathbf{x}_j^{\prime} \in \mathbb{R}:\left(\mathbf{a}^{\prime}, \mathbf{x}^{\prime}\right) \in E_{\epsilon, \Omega}(F)\right\} \equiv P_{\mathbf{x}, j} E_{\epsilon, \Omega}(F), \\
	&E_{\epsilon, \Omega}^{\mathbf{a}, j}(F)=\left\{\mathbf{a}_j^{\prime} \in \mathbb{C}:\left(\mathbf{a}^{\prime}, \mathbf{x}^{\prime}\right) \in E_{\epsilon, \Omega}(F)\right\} \equiv P_{\mathbf{a}_{,} j} E_{\epsilon, \Omega}(F) .
\end{aligned}
$$

For any subset $V$ of a normed vector space with norm $\|\cdot\|$, the diameter of $V$ is given by
$$
\mathrm{diam}(V)=\sup _{\mathbf{v}^{\prime}, \mathbf{v}^{\prime \prime} \in V}\left\|\mathbf{v}^{\prime}-\mathbf{v}^{\prime \prime}\right\|_{\infty} .
$$

By the theory of optimal recovery \cite{micchelli1977survey, micchelli1985lectures, micchelli1976optimal}, the minimax errors are directly linked to the diameter of the corresponding projections of the error set. More specifically, we have the following proposition that is similar to Proposition 2.4 in  \cite{batenkov2019super}. 

\begin{prop}\label{prop:diameterbound1}
For $U \subset \mathcal{P}_d^+, \Omega>0,1 \leqslant j \leqslant d$ and $\epsilon>0$, we have
\begin{equation}
\begin{aligned}
&\frac{1}{2} \sup _{F: E_{\frac{1}{2} \epsilon, \Omega}^+(F) \subseteq U} \mathrm{diam}\left(E_{\frac{1}{2} \epsilon, \Omega}^+(F)\right) \leqslant \mathcal{E}^+(\epsilon, U, \Omega) \leqslant \sup _{F \in U} \mathrm{diam}\left(E_{2 \epsilon, \Omega}^+(F)\right), \\
&\frac{1}{2} \sup _{F: E_{\frac{1}{2} \epsilon, \Omega}^+(F) \subseteq U} \mathrm{diam}\left(E_{\frac{1}{2} \epsilon, \Omega}^{+, \mathbf{x}, j}(F)\right) \leqslant \mathcal{E}^{+,\mathbf{x}, j}(\epsilon, U, \Omega) \leqslant \sup _{F \in U} \mathrm{diam}\left(E_{2 \epsilon, \Omega}^{+,\mathbf{x}, j}(F)\right), \\
&\frac{1}{2} \sup _{F:E_{\frac{1}{2} \epsilon, \Omega}^+(F) \subseteq U} \mathrm{diam}\left(E_{\frac{1}{2} \epsilon, \Omega}^{+,\mathbf{a}, j}(F)\right) \leqslant \mathcal{E}^{+, \mathbf{a}, j}(\epsilon, U, \Omega) \leqslant \sup _{F \in U} \mathrm{diam}\left(E_{2 \epsilon, \Omega}^{+,\mathbf{a}, j}(F)\right).
\end{aligned}
\end{equation}
\end{prop}

\subsection{Uniform estimates of minimax error for clustered configurations}
Similarly to \cite{batenkov2019super}, the main goal of this paper is to estimate $\mathcal{E}^{+,\mathbf{x}, j}(\epsilon, U, \Omega), \mathcal{E}^{+,\mathbf{a}, j}(\epsilon, U, \Omega)$, where $U \subset \mathcal{P}_d^+$ are certain compact subsets of  $\mathcal{P}_d^+$ consisting of signals with $p \leqslant d$ nodes that are nearly uniformly distributed, forming a cluster. To be more specific, the set $U$ is defined as follows; See also \cite[Definition 2.5]{batenkov2019super}.

\begin{defi}\label{defi:clusterconfig}(\textbf{Uniform cluster configuration})\\
Given $0<\tau, \eta \leqslant 1$ and $0<h \leqslant T$, a node vector $\mathbf{x}=\left(x_1, \ldots, x_d\right) \in \mathbb{R}^d$ is said to form a $(p, h, T, \tau, \eta)$-clustered configuration, if there exists a subset of $p$ nodes $\mathbf{x}^c=\left\{x_\kappa, \ldots, x_{\kappa+p-1}\right\} \subset \mathbf{x}, p \geqslant 2$, which satisfies the following conditions:
\begin{itemize}
\item[(i)] for each $x_j, x_k \in \mathbf{x}^c, j \neq k$,
\[
\tau h \leqslant\left|x_j-x_k\right| \leqslant h ;
\]
\item [(ii)] for $x_{\ell} \in \mathbf{x} \backslash \mathbf{x}^c$ and $x_j \in \mathbf{x}, \ell \neq j$,
\[
\eta T \leqslant\left|x_{\ell}-x_j\right| \leqslant T .
\]
\end{itemize}
\end{defi}

One of the main contributions of \cite{batenkov2019super} is an upper bound on $\operatorname{diam}\left(E_{\epsilon, \Omega}(F)\right)$, and its coordinate projections, for any signal $F$ forming a clustered configuration as above. Here, we generalize the result to the positive signal cases, which is a direct consequence of \cite[Theorem 2.6]{batenkov2019super}.

\begin{thm}\label{thm:diameterupperbound}(\textbf{Upper bound})\\
Let the positive signal $F=(\mathbf{a}, \mathbf{x}) \in \mathcal{P}_d^+$, such that $\mathbf{x}$ forms a $(p, h, T, \tau, \eta)$-clustered configuration and $0<m \leqslant\|\mathbf{a}\|$. Then, there exist positive constants $C_1, \ldots, C_5$, depending only on $d, p, m$, such that for each $\frac{C_4}{\eta T} \leqslant \Omega \leqslant \frac{C_5}{h}$ and $\epsilon \leqslant C_3(\Omega \tau h)^{2 p-1}$, it holds that
\[
\begin{gathered}
	\operatorname{diam}\left(E_{\epsilon, \Omega}^{+, \mathbf{x}, j}(F)\right) \leqslant \frac{C_1}{\Omega} \epsilon \times \begin{cases}(\Omega \tau h)^{-2 p+2}, & x_j \in \mathbf{x}^c, \\
		1, & x_j \in \mathbf{x} \backslash \mathbf{x}^c,\end{cases} \\
	\operatorname{diam}\left(E_{\epsilon, \Omega}^{+, \mathbf{a}, j}(F)\right) \leqslant C_2 \epsilon \times \begin{cases}(\Omega \tau h)^{-2 p+1}, & x_j \in \mathbf{x}^c, \\
		1, & x_j \in \mathbf{x} \backslash \mathbf{x}^c .\end{cases}
\end{gathered}
\]
\end{thm}
\begin{proof}
	By Theorem 2.6 in \cite{batenkov2019super}, under the same condition we have 
	\[
	\begin{gathered}
		\operatorname{diam}\left(E_{\epsilon, \Omega}^{\mathbf{x}, j}(F)\right) \leqslant \frac{C_1}{\Omega} \epsilon \times \begin{cases}(\Omega \tau h)^{-2 p+2}, & x_j \in \mathbf{x}^c, \\
			1, & x_j \in \mathbf{x} \backslash \mathbf{x}^c,\end{cases} \\
		\operatorname{diam}\left(E_{\epsilon, \Omega}^{\mathbf{a}, j}(F)\right) \leqslant C_2 \epsilon \times \begin{cases}(\Omega \tau h)^{-2 p+1}, & x_j \in \mathbf{x}^c, \\
			1, & x_j \in \mathbf{x} \backslash \mathbf{x}^c .\end{cases}
	\end{gathered}
	\]
	On the other hand, note that for $\epsilon>0$, according to the definition of $\mathrm{diam}(\cdot)$ we have 
	\[
	\mathrm{diam}\left(E_{\epsilon,\Omega}^{+, \vect x, j}(F)\right) \leq \mathrm{diam}\left(E_{\epsilon,\Omega}^{\vect x, j}(F)\right), \quad  \mathrm{diam}\left(E_{\epsilon,\Omega}^{+, \vect a, j}(F)\right) \leq \mathrm{diam}\left(E_{\epsilon,\Omega}^{\vect a, j}(F)\right).
	\]
   This proves the theorem.
\end{proof}

The above estimates are optimal, as shown by our next main theorem. This is the main contribution of our paper, by which we non-trivially generalize the results in \cite[Theorem 2.7]{batenkov2019super}. For simplicity and without loss of generality, we assume that the index $\kappa$ is fixed in the result below.

\begin{thm}\label{thm:diameterlowerbound}(\textbf{Lower bound})\\ 
Let $m \leqslant M, 2 \leqslant p \leqslant d, \tau \leqslant \frac{1}{p-1}, \eta<\frac{1}{d}, T>0$ be fixed. There exist positive constants $C_1^{\prime} \ldots, C_5^{\prime}$, depending only on $d, p, m, M$, such that for every $\Omega, h$ satisfying $h \leqslant C_4^{\prime} T$ and $\Omega h \leqslant C_5^{\prime}$, there exists $F=(\mathbf{a}, \mathbf{x}) \in\mathcal {P}_d^+$, with $\mathbf{x}$ forming a $(p, h, T, \tau, \eta)$-clustered configuration, and with $0<m \leqslant\|\mathbf{a}\| \leq M< \infty$, such that for certain indices \( j_1, j_2 \in\{\kappa, \ldots, \kappa+p-1\} \) and every \( \epsilon \leqslant C_3^{\prime}(\Omega \tau h)^{2 p-1} \), it holds that
$$
\begin{aligned}
&\operatorname{diam}\left(E_{\epsilon, \Omega}^{+, \mathbf{x}, j}(F)\right) \geqslant \frac{C_1^{\prime}}{\Omega} \epsilon \times \begin{cases}(\Omega \tau h)^{-2 p+2}, &  j = j_1, \\
1, & \forall j \notin\{\kappa, \ldots, \kappa+p-1\},\end{cases} \\
&\operatorname{diam}\left(E_{\epsilon, \Omega}^{+, \mathbf{a}, j}(F)\right) \geqslant C_2^{\prime} \epsilon \times \begin{cases}(\Omega \tau h)^{-2 p+1}, & j =j_2, \\
1, & \forall j \notin\{\kappa, \ldots, \kappa+p-1\}.\end{cases}
\end{aligned}
$$
\end{thm}
\begin{proof}
See the proof in Section \ref{section:proofoflowerbound}. 
\end{proof}

Thus combining Theorems \ref{thm:diameterupperbound} and \ref{thm:diameterlowerbound} with Proposition \ref{prop:diameterbound1} now, we can obtain the following theorem for the optimal rates of the minimax errors $\mathcal{E}^{+, \vect x, j}, \mathcal{E}^{+, \vect a, j}$. This generalizes Theorem 2.8 in \cite{batenkov2019super}. 

\begin{thm}\label{thm:minimax1}
Let $m>0, 2 \leqslant p \leqslant d, \tau<\frac{1}{2(p-1)}, \eta<\frac{1}{2 d}, T>0$ be fixed. There exist constants $c_1, c_2, c_3$, depending only on $d, p, m$ such that for all $\frac{c_1}{\eta T} \leqslant \Omega \leqslant \frac{c_2}{h}$ and $\epsilon \leqslant c_3(\Omega \tau h)^{2 p-1}$, the minimax error rates for the set
\begin{align*}
U:=&U(p, d, h, \tau, \eta, T, m)\\
=&\left\{(\mathbf{a}, \mathbf{x}) \in\mathcal {P}_d^+ :  
0<m \leqslant\|\mathbf{a}\| \leq M<\infty,\text{$\mathbf{x}$ forms a $(p, h, T, \tau, \eta)$-clustered configuration} \right\}
\end{align*}
satisfy the following:
\begin{itemize}
	\item For the non-cluster nodes:
$$
\forall j \notin\{\kappa, \ldots, \kappa+p-1\}: \quad\left\{\begin{array}{l}
\mathcal {E}^{+,\mathbf{x}, j}(\epsilon, U, \Omega) \asymp \frac{\epsilon}{\Omega}, \\
\mathcal {E}^{+, \mathbf{a}, j}(\epsilon, U, \Omega) \asymp \epsilon .
\end{array}\right.
$$
\item For the cluster nodes:
$$
\begin{aligned}
\max _{j=\kappa, \ldots, \kappa+p-1}\mathcal {E}^{+, \mathbf{x}, j}(\epsilon, U, \Omega) & \asymp \frac{\epsilon}{\Omega}(\Omega \tau h)^{-2 p+2}, \\
\max _{j=\kappa, \ldots, \kappa+p-1}\mathcal {E}^{+, \mathbf{a}, j}(\epsilon, U, \Omega) & \asymp \epsilon(\Omega \tau h)^{-2 p+1} .
\end{aligned}
$$
\end{itemize}
The proportionality constants in the above statements depend only on $d, p, m, M$.
\end{thm}
\begin{proof}
The proof is the same as the one for Theorem 2.8 in \cite{batenkov2019super}. Here, we present the details for the convenience of reading and completeness. 
Let $C_3, C_3^{\prime}, C_4, C_4^{\prime}, C_5, C_5^{\prime}$ be the constants from Theorems \ref{thm:diameterupperbound} and \ref{thm:diameterlowerbound}. Let $c_1=C_4$ and $c_2=\min \left(C_5, C_5^{\prime}, C_4 C_4^{\prime}\right)$. Let $\frac{c_1}{\eta T} \leqslant \Omega \leqslant \frac{c_2}{h}$, and $\epsilon \leqslant c_3(\Omega \tau h)^{2 p-1}$, where $c_3 \leqslant \min \left(C_3, C_3^{\prime}\right)$ will be determined below. It can be verified that $\Omega, h$ and $\epsilon$ as above satisfy the conditions of both Theorems \ref{thm:diameterupperbound} and \ref{thm:diameterlowerbound}.

\medskip
\noindent \textbf{Upper bound} Directly follows from the upper bounds in Theorem \ref{thm:diameterupperbound} and Proposition \ref{prop:diameterbound1}. 

\medskip
\noindent \textbf{Lower bound} Denote $U_\epsilon=\left\{F \in U: E_{\frac{1}{2} \epsilon, \Omega}^+(F) \subseteq U\right\}$. In order to prove the lower bounds on $\mathcal {E}^{+, \mathbf{x}, j}$ and $\mathcal E^{+, \mathbf{a}, j}$, by Proposition \ref{prop:diameterbound1} it suffices to show that there exists an $F \in U_\epsilon \neq \emptyset$ such that the conclusions of Theorem \ref{thm:diameterlowerbound} are satisfied for this $F$. Note that the set $U$ has a non-empty interior. Furthermore, one can choose $m^{\prime}, M^{\prime}$ satisfying $m<m^{\prime}<M^{\prime}<M$, and also $T^{\prime}=0.99 T, \tau^{\prime}=2 \tau$ and $\eta^{\prime}=2 \eta$, such that
\[
U^{\prime}=U\left(p, d, h, \tau^{\prime}, \eta^{\prime}, T^{\prime}, m^{\prime}, M^{\prime}\right) \subset U, \quad \partial U^{\prime} \cap \partial U=\emptyset.
\]
By the construction of $U^{\prime}$, there exist positive constants $\widetilde{C}_1, \widetilde{C}_2$, independent of $\Omega, h$ and $\tau, \eta$, such that
\begin{equation}\label{equ:proofminimax1}
\begin{aligned}
&\inf _{u \in \partial U, u^{\prime} \in \partial U^{\prime}}\left|P_{\mathbf{x}, j}(u)-P_{\mathbf{x}, j}\left(u^{\prime}\right)\right| \geqslant \tilde{C}_1 \times \begin{cases}\tau h, & x_j \in \mathbf{x}^c, \\ \eta T, & x_j \in \mathbf{x} \backslash \mathbf{x}^c ;\end{cases}\\
&\inf _{u \in \partial U, u^{\prime} \in \partial U^{\prime}}\left|P_{\mathbf{a}, j}(u)-P_{\mathbf{a}, j}\left(u^{\prime}\right)\right| \geqslant \tilde{C}_2.
\end{aligned}
\end{equation}
Now, we use the fact that $\epsilon<c_3(\Omega \tau h)^{2 p-1}$. Applying Theorem \ref{thm:diameterupperbound} to an arbitrary positive signal $F^{\prime} \in U^{\prime}$ and using the conditions $\frac{1}{\Omega} \leqslant \frac{\eta T}{c_1}$ and $\Omega \tau h \leqslant \Omega h \leqslant c_2$, we obtain that
\begin{equation}\label{equ:proofminimax2}
\begin{aligned}
\operatorname{diam}\left(E_{\frac{1}{2} \epsilon}^{+, \mathbf{x}, j}\left(F^{\prime}\right)\right) & \leqslant \begin{cases}\frac{C_1 c_3}{2} \tau h, & x_j \in \mathbf{x}^c, \\
\frac{C_1 c_3}{2 \Omega}(\Omega \tau h)^{2 p-1} \leqslant \frac{C_1 c_3}{2 c_1} c_2^{2 p-1} \eta T, & x_j \in \mathbf{x} \backslash \mathbf{x}^c ;\end{cases} \\
\operatorname{diam}\left(E_{\frac{1}{2} \epsilon}^{+, \mathbf{a}, j}\left(F^{\prime}\right)\right) & \leqslant \begin{cases}\frac{C_2 c_3}{2}, & x_j \in \mathbf{x}^c, \\
\frac{C_2 c_3}{2} c_2^{2 p-1}, & x_j \in \mathbf{x} \backslash \mathbf{x}^c .\end{cases}
\end{aligned}
\end{equation}
Next, we set $c_3=\min \left(C_3, C_3^{\prime}, C_3^{\prime \prime}\right)$, where
$$
C_3^{\prime \prime}=\min \left(1, c_1\right) \times \min \left(1, c_2^{-2 p+1}\right) \times \min \left(\frac{2 \tilde{C}_1}{C_1}, \frac{2 \tilde{C}_2}{C_2}\right) .
$$
Combining (\ref{equ:proofminimax1}) and (\ref{equ:proofminimax2}), we obtain that $ E_{\frac{1}{2} \epsilon, \Omega}^+(F^{\prime}) \subseteq U$, i.e. $F^{\prime} \in U_\epsilon$. Since $F^{\prime} \in U^{\prime}$ is arbitrary, we conclude that $U^{\prime} \subseteq U_\epsilon$. Since clearly $U^{\prime} \neq \emptyset$, applying Proposition \ref{prop:diameterbound1} and Theorem \ref{thm:diameterlowerbound} finishes the proof.
\end{proof}

\section{Numerical optimality of Matrix Pencil method (MP method)}
Theorem \ref{thm:minimax1} establishes the optimal error rate for super-resolving the locations and amplitudes of positive sources. In this section, we demonstrate by numerical experiments the optimal performance of MP method in recovering the locations of positive sources. Note that the numerical experiments in \cite{batenkov2019super} have already demonstrated the optimal performance of Matrix Pencil method for resolving general sources. Here we conduct experiments similar to those in \cite{batenkov2019super} but focusing on the case of resolving positive sources. 


\subsection{Review of the MP method}
In this section, we review the MP method . We assume that the noisy Fourier data of the signal $F$ is given by
\[
\vect Y(\omega) = \sum_{j=1}^{d} a_j e^{-2\pi i x_j \omega}+ \epsilon(\omega), \quad \omega \in [-\Omega, \Omega]. 
\] 
The measurements are usually taken at $N$ evenly spaced points, $\omega_1=-\Omega, \omega_2=-\Omega+h,  \cdots, \omega_{N}=\Omega$ with $h$ being the spacing.  From the measurement 
\begin{equation}\label{equ:discretemeasurement}
	\mathbf Y =(\mathbf Y(\omega_1),\mathbf Y(\omega_2),\cdots,\mathbf Y(\omega_N))^\top
	\end{equation}
and $\hat N=\lfloor \frac{N-1}{2}\rfloor$, we assemble the $(\hat N +1)\times(\hat N +1)$ Hankel matrix
\begin{equation}\label{equ:hankelmatrix1}\vect H=\begin{pmatrix}
		\mathbf Y(\omega_1)&\mathbf Y(\omega_2)&\cdots& \mathbf Y(\omega_{\hat N})\\
		\mathbf Y(\omega_2)&\mathbf Y(\omega_3)&\cdots&\mathbf Y(\omega_{\hat N+1})\\
		\cdots&\cdots&\ddots&\cdots\\
		\mathbf Y(\omega_{\hat N})&\mathbf Y(\omega_{\hat N+1})&\cdots&\mathbf Y(\omega_{2\hat N+1})
	\end{pmatrix}.
\end{equation}

Let $\vect H_u:= \vect H[1:\hat N, :]$ (and $H_l := H[2:\hat N+1, :]$) be the $\hat N\times (\hat N+1)$ matrix obtained from the Hankel matrix $\vect H$ given by (\ref{equ:hankelmatrix1}) by selecting the first $\hat N$ rows (respectively, the second to the $(\hat N+1)$-th rows). It turns out that, in the noiseless case,
$e^{-2\pi i x_j h}, 1\leq j \leq n$,  are exactly the nonzero generalized eigenvalues of the
pencil $\vect H_{l}-z\vect H_{u}$. In the noisy case, when the sources are well-separated, each of the first $n$ nonzero generalized eigenvalues of the
pencil $\vect H_{l}-z\vect H_{u}$ is close to $e^{-2\pi i  x_j h}$ for some $j$. We summarize the Matrix Pencil method in \textbf{Algorithm \ref{algo:onedmpsupport}}. 

\begin{algorithm}[!h]
	\caption{\textbf{The Matrix Pencil algorithm}}\label{algo:onedmpsupport}	
	\textbf{Input:} Source number $d$, measurement: $\mathbf Y$ in (\ref{equ:discretemeasurement});\\
	1: Let $\hat N= \lfloor \frac{N-1}{2}\rfloor$. Formulate the $(\hat N+1)\times(\hat N+1)$ Hankel matrix $\mathbf H$ given by (\ref{equ:hankelmatrix1}) from $\mathbf{Y}$, and the matrices $\vect H_u, \vect H_l$\;
	2: Compute the truncated Singular Value Decomposition (SVD) of $\vect H_u$, $\vect H_l$ of order $d$:
	\[
	\vect H_u=U_1\Sigma_1V_1^*,\quad \vect H_l = U_2\Sigma_2V_2^*,
	\]
	where $U_1, U_2, V_1, V_2$ are $\hat N\times d$ matrices and $\Sigma_1, \Sigma_2$ are $d \times d$\ matrices\;
	3: Generate the reduced pencil
	\[
	\mathbf {\hat H}_u =  U_2^* U_1\Sigma_1V_1^* V_2,\quad  \mathbf {\hat H}_l= \Sigma_2,
	\]
	where $\mathbf {\hat H}_u$, $\mathbf {\hat H}_l$ are $d\times d$ matrices\;
	4: Compute the generalized eigenvalues $\{\hat z_j\}$ of the reduced pencil $(\mathbf {\hat H}_u, \mathbf {\hat H}_l)$, and put
	$\{\hat x_j\}=\{\arg(\hat z_j)\}, j=1, \cdots, n$, where the $\arg(z)$ is the argument of $z$\;
	5: Solving the linear least squares problem
	\[
	\mathbf{\hat b} =\arg \min _{\mathbf{b} \in \mathbb{C}^d}\bnorm{\vect Y-V \mathbf{b}}_2,
	\]
	where $V$ is the Vandermonde matrix $V=\left(e^{-2 \pi i \hat{x}_j \omega_k}\right)_{k=1, \cdots, N}^{j=1, \cdots, d}$\;
	6:  Compute $\hat{a}_j$ by $|\mathbf{\hat b}_j|$\; 
	\textbf{Return:} The estimated $\hat x_j$'s and $\hat a_j$'s.
\end{algorithm}

\subsection{Numerical experiments}
We conduct 1000 random experiments (the randomness was in the choice of $a_j, x_j, \epsilon$) to examine the error amplification in the recovery of the nodes and amplitudes. In particular, we consider recovering $p=2$ cluster nodes and $1$ non-cluster node and their corresponding amplitudes. Each single experiment is summarized in \textbf{Algorithm \ref{algo:singleexperiment}}.  The results are shown in Figure \ref{fig:mplocaamplitude} and we observe that the error amplification is consistent with what we have predicted,  i.e., the error amplification factors for resolving nodes and amplitudes are respectively $\mathrm{SRF}^{2p-2}$ and $\mathrm{SRF}^{2p-1}$ for the cluster nodes with size $p$. Moreover, for resolving non-cluster nodes, both the corresponding error amplification factors are bounded by a small constant.

\begin{algorithm}[H]
	\caption{\textbf{A single experiment}}\label{algo:singleexperiment}	
\textbf{Input:} $p, d, N, \epsilon$;\\
1: Construct the signal $F$ with $p$ closely-spaced sources and $d-p$ non-clustered sources\;
2: Generate the measurement $\vect Y$ defined by (\ref{equ:discretemeasurement}) with $\epsilon$ being the noise level\;
3: Execute the MP method (\textbf{Algorithm \ref{algo:onedmpsupport}}) and obtain $F_{M P}=\left(\mathbf{a}^{MP}, \mathbf{x}^{MP}\right)$. The nodes in $\mathbf{x}^{MP}$ are ordered in an increasing manner\;
4: \For{each $j$}{
Compute the error for node $j$ :
$$
e_j=\left|\mathbf{x}_j^{MP}-\mathbf{x}_{j}\right|. 
$$
The success for node $j$ is defined as
$$
\operatorname{Succ}_j=\left(e_j<\frac{\min _{\ell \neq j}\left|\mathbf{x}_{\ell}-\mathbf{x}_j\right|}{3}\right) .
$$
 \If{Succ $_j==$ true} { 
 Compute normalized node error amplification factor
$$
\mathcal{K}_{\mathbf{x}, j}=\frac{\left|\mathbf{x}_j-\mathbf{x}_{j}^{MP}\right| \cdot \Omega}{\epsilon};
$$
Compute normalized amplitude error amplification factor
$$
\mathcal{K}_{\mathbf{a}, j}=\frac{\left|\mathbf{a}_j-\mathbf{a}_{j}^{MP}\right|}{\epsilon} ;
$$}}
\textbf{Return:} $\left(\mathcal{K}_{\mathbf{x}, j}, \mathcal{K}_{\mathbf{a}, j}, \operatorname{Succ}_j\right)$ for each node $j=1, \ldots, d$.
\end{algorithm}

\begin{figure}[!h]
	\centering
	\begin{subfigure}[b]{0.48\textwidth}
		\centering
		\includegraphics[width=\textwidth]{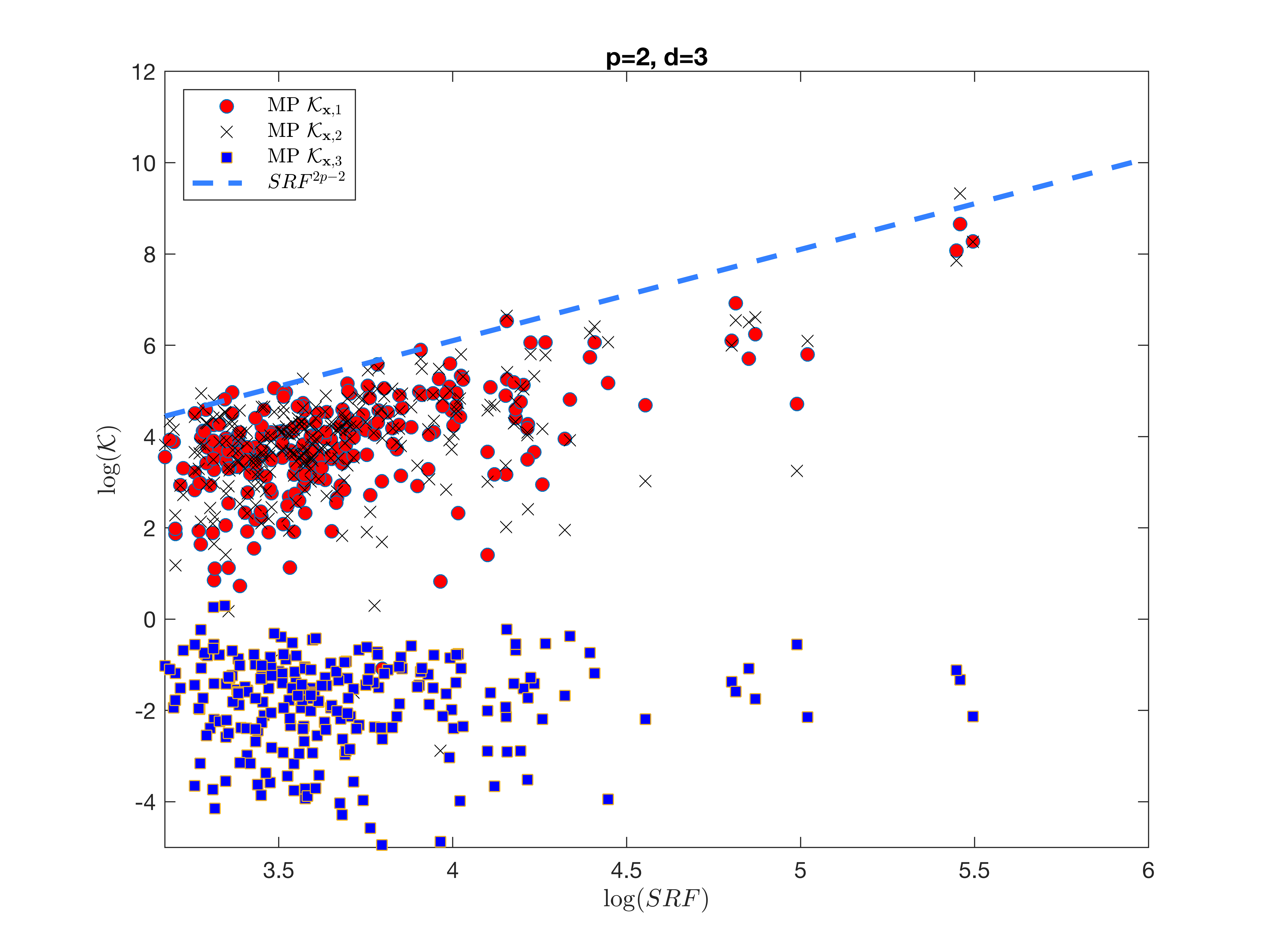}
		\caption{}
	\end{subfigure}
	\begin{subfigure}[b]{0.48\textwidth}
		\centering
		\includegraphics[width=\textwidth]{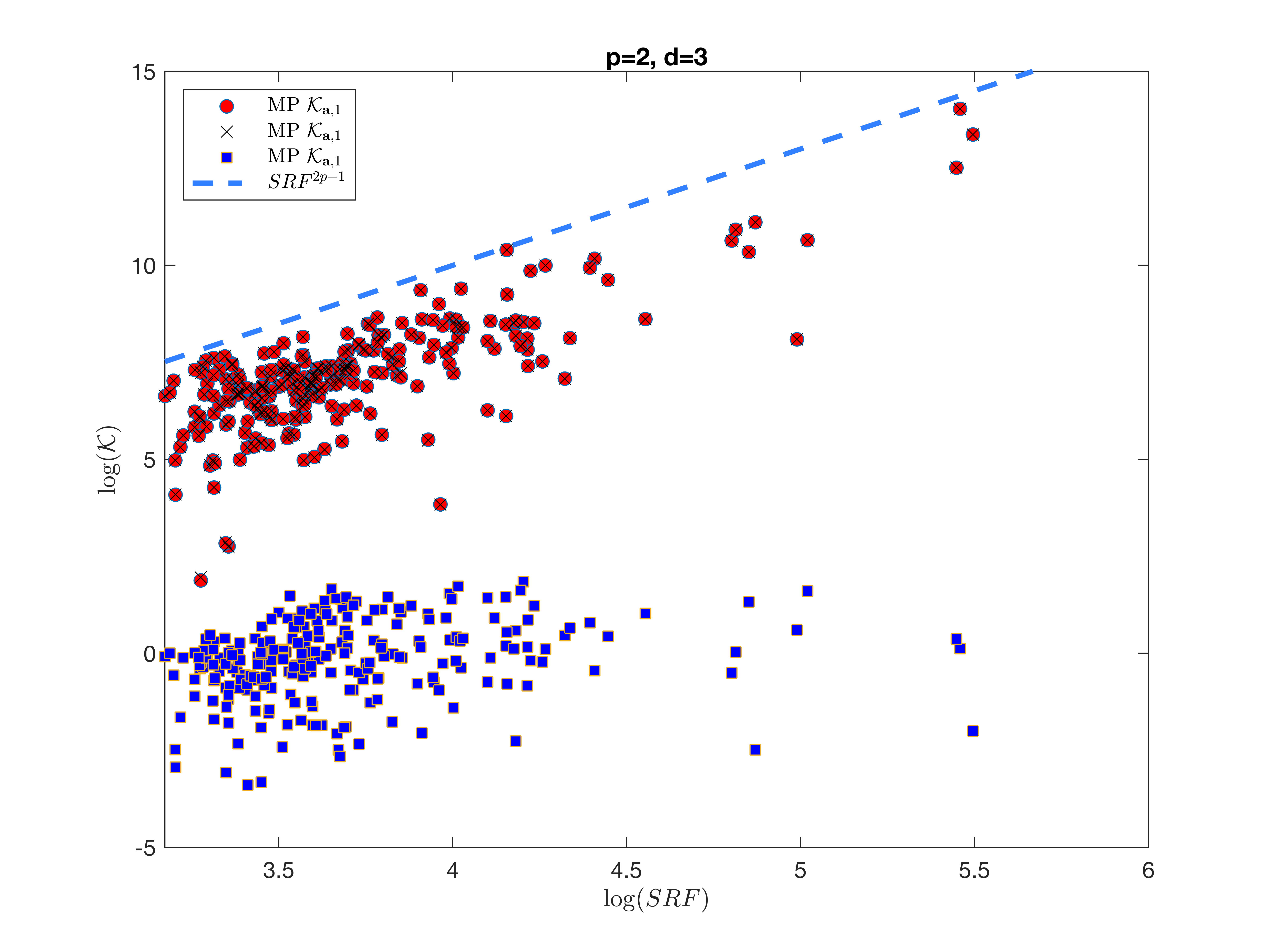}
		\caption{}
	\end{subfigure}
	\caption{The error amplification factors. For the cluster nodes, the error amplification factors $\mathcal{K}_{\vect x, j}, \mathcal{K}_{\vect a, j}$ scale like $\mathrm{SRF}^{2p-2}, \mathrm{SRF}^{2p-1}$, respectively. For the non-cluster nodes, both error amplification factors are bounded by a small constant. }
	\label{fig:mplocaamplitude}
\end{figure}

\section{Proof of Theorem \ref{thm:diameterlowerbound}}\label{section:proofoflowerbound}

\subsection{Normalization}
Similarly to \cite{batenkov2019super}, for ease of exposition, we should normalize the cluster configuration in some of the following discussions. Let us first define the scale transformation on $\mathcal{P}_d^+$.
\begin{defi}\label{defi:scaling}
For $F=\sum_{j=1}^d a_j \delta\left(x-x_j\right) \in \mathcal{P}_d^+$ and $T>0$, we define $S C_T: \mathcal{P}_d^+ \rightarrow \mathcal{P}_d^+$ as follows:
$$
S C_T(F)(x)=\sum_{j=1}^d a_j \delta\left(x-\frac{x_j}{T}\right).
$$
\end{defi}
By the scale property of the Fourier transform, we have that for any $\epsilon>0$,
\[
SC_T\left(E_{\epsilon, \Omega}^+(F)\right)=E_{\epsilon, \Omega T}^+\left(S C_T(F)\right) .
\]
Thus the following proposition holds. 
\begin{prop}\label{prop:scaleidentity}
Let $F=(\mathbf{a}, \mathbf{x}) \in \mathcal{P}_d^+$ and $T>0$. Then, for any $\epsilon>0$ and $1 \leqslant j \leqslant d$, we have
\begin{align*}
&\operatorname{diam}\left(E_{\epsilon, \Omega}^{+,\mathbf{x}, j}(F)\right)=T \operatorname{diam}\left(E_{\epsilon, \Omega T}^{+,\mathbf{x}, j}\left(S C_T\left(F\right)\right)\right),\\
&\operatorname{diam}\left(E_{\epsilon, \Omega}^{+,\mathbf{a}, j}(F)\right)=\operatorname{diam}\left(E_{\epsilon, \Omega T}^{+,\mathbf{a}, j}\left(S C_T\left(F\right)\right)\right). 
\end{align*}
\end{prop}

\subsection{Auxiliary lemmas}
In this subsection, we introduce some notation and lemmas that are used in the following proofs. Set 
\begin{equation}\label{equ:phiformula}
	\phi_{s}(t) = \left(1, t, \cdots, t^s\right)^{\top}.
\end{equation}
\begin{lem}{\label{lem:invervandermonde}}
	Let $t_1, \cdots, t_k$ be $k$ different real numbers and let $t$ be a real number. We have
	\[
	\left(D_k(k-1)^{-1}\phi_{k-1}(t)\right)_{j}=\Pi_{1\leq q\leq k,q\neq j}\frac{t- t_q}{t_j- t_q},
	\]
	where $D_k(k-1):=  \big(\phi_{k-1}(t_1),\cdots,\phi_{k-1}(t_k)\big)$ with $\phi_{k-1}(\cdot)$ defined by (\ref{equ:phiformula}). 
\end{lem}
\begin{proof}
	This is \cite[Lemma 5]{liu2021theorylse}. 
\end{proof}

The following proposition is the main result for proving Theorem \ref{thm:diameterlowerbound} and we present its detailed proof in Section \ref{section:proofproplower}. 
\begin{prop}\label{prop:diameterlowerestimate1}
Let $F=(\mathbf{a}, \mathbf{x}) \in \mathcal{P}_d^+$, such that $\mathbf{x}$ forms a $(p, h, 1, \tau, \eta)$-clustered configuration, with cluster nodes $\mathbf{x}^c=\left(x_1, \ldots, x_p\right)$ (according to Definition \ref{defi:clusterconfig}), and with $\mathbf{a} \in (\mathbb{R}^+)^d$ satisfying $m \leqslant\|\mathbf{a}\| \leqslant M$.
Then, there exist constants $c_1, k_1, k_2, k_3, k_4$, depending only on $(d, p, \tau, m, M)$, such that for all $\epsilon<c_1(\Omega h)^{2 p-1}$ and $\Omega h \leqslant 2$, there exists a signal $F_\epsilon \in \mathcal{P}_d^+$ satisfying, for some $j_1, j_2 \in\{1, \ldots, p\},$
\begin{align*}
&\left|P_{\mathbf{x}, j_1}\left(F_\epsilon\right)-P_{\mathbf{x}, j_1}(F)\right| \geqslant \frac{k_1}{\Omega}(\Omega h)^{-2 p+2} \epsilon,\\
&\left|P_{\mathbf{a}, j_2}\left(F_\epsilon\right)-P_{\mathbf{a}, j_2}(F)\right| \geqslant k_2(\Omega h)^{-2 p+1} \epsilon,\\
&\left|P_{\mathbf{x}, j}\left(F_\epsilon\right)-P_{\mathbf{x}, j}(F)\right| \geqslant \frac{k_3}{\Omega} \epsilon,\quad  x_j \in \mathbf{x} \backslash \mathbf{x}^c, \\
&\left|P_{\mathbf{a}, j}\left(F_\epsilon\right)-P_{\mathbf{a}, j}(F)\right| \geqslant k_4 \epsilon,\quad  x_j \in \mathbf{x} \backslash \mathbf{x}^c, \\
&\left|\mathcal{F}\left(F_\epsilon\right)(s)-\mathcal{F}(F)(s)\right| \leqslant \epsilon, \quad|s| \leqslant \Omega .
\end{align*}
\end{prop}

\subsection{Proof of Theorem \ref{thm:diameterlowerbound}}
\begin{proof}
After employing Proposition \ref{prop:diameterlowerestimate1}, the arguments for proving Theorem \ref{thm:diameterlowerbound} is just the same as those in \cite{batenkov2019super}. We present the details as follows. 
 
Let $\mathbf{a} \in (\mathbb{R}^+)^d$ be any positive amplitude vector satisfying $m \leqslant\|\mathbf{a}\| \leqslant M$. Let $\Omega$, $h$ satisfy $\Omega h \leqslant 2$, and choose nodes $\mathbf{x}$ satisfying that
$$
\mathbf{x}^c=\left(x_1=0, x_1=\tau h, \ldots, x_p=(p-1) \tau h\right),
$$
and the rest of the non-cluster nodes are equally spaced in $((p-1) \tau h, 1)$. Now, let $h^{\prime}=(p-1) \tau h$ and $\tau^{\prime}=\frac{1}{p-1}$. Clearly, $\mathbf{x}$ is a $\left(p, h^{\prime}, 1, \tau^{\prime}, \eta\right)$-clustered configuration for all sufficiently small $h$ (for instance $h<\frac{1}{d}<1-\eta(d-p+1)$).We now can apply Proposition \ref{prop:diameterlowerestimate1} to the signal $F=(\mathbf{a}, \mathbf{x})$. It then follows that for $\epsilon<c_1(p-1)^{2 p-1}(\Omega \tau h)^{2 p-1}$ and $\Omega h<\frac{2}{(p-1) \tau}$, there exist $j_1, j_2 \in\{1, \ldots, p\}$ such that
\[
\begin{aligned}
	\operatorname{diam}\left(E_{\epsilon, \Omega}^{+,\mathbf{x}, j_1}(F)\right) & \geqslant \frac{k_1}{\Omega}(p-1)^{-2 p+2} \epsilon(\Omega \tau h)^{-2 p+2}, \\
	\operatorname{diam}\left(E_{\epsilon, \Omega}^{+,\mathbf{a}, j_2}(F)\right) & \geqslant k_2 \epsilon(p-1)^{-2 p+1}(\Omega \tau h)^{-2 p+1} .
\end{aligned}
\]
Moreover, 
\[
\begin{aligned}
	\operatorname{diam}\left(E_{\epsilon, \Omega}^{+,\mathbf{x}, j}(F)\right) & \geqslant \frac{k_3}{\Omega} \epsilon, \quad x_j \in \mathbf{x} \backslash \mathbf{x}^c, \\
	\operatorname{diam}\left(E_{\epsilon, \Omega}^{+,\mathbf{a}, j}(F)\right) & \geqslant k_4 \epsilon, \quad x_j \in \mathbf{x} \backslash \mathbf{x}^c.
\end{aligned}
\]

For the general case that $F=(\vect a, \vect x)\in \mathcal P_{d}^+$ such that $\mathbf{x}$ forms a $(p, h, T, \tau, \eta)$-clustered configuration, we consider $SC_{T}(F)=(\vect a, \mathbf{\tilde{x}})$, $\mathbf{\tilde{x}}=(\tilde{x}_1, \cdots, \tilde{x}_d)$, where $\tilde{x}_j=\frac{x_j}{T}, j=1, \cdots, d$. Now the node vector $\mathbf{\tilde{x}}$ forms a $(p, \frac{h}{T}, 1, \tau, \eta)$-clustered configuration. Applying Proposition \ref{prop:scaleidentity} and the above results, we obtain that
\[
\begin{aligned}
	&\operatorname{diam}\left(E_{\epsilon, \Omega}^{+,\mathbf{x}, j_1}(F)\right)= T\operatorname{diam}\left(E_{\epsilon, \Omega T}^{+,\mathbf{x}, j_1}(SC_T(F))\right)  \geqslant \frac{k_1}{\Omega}(p-1)^{-2 p+2} \epsilon(\Omega \tau h)^{-2 p+2}, \\
	&\operatorname{diam}\left(E_{\epsilon, \Omega}^{+,\mathbf{a}, j_2}(F)\right)= \operatorname{diam}\left(E_{\epsilon, \Omega T}^{+,\mathbf{a}, j_2}(SC_T(F))\right)  \geqslant k_2 \epsilon(p-1)^{-2 p+1}(\Omega \tau h)^{-2 p+1},\\
	&\operatorname{diam}\left(E_{\epsilon, \Omega}^{+,\mathbf{x}, j}(F)\right)= T\operatorname{diam}\left(E_{\epsilon, \Omega T}^{+,\mathbf{x}, j}(SC_T(F))\right)  \geqslant \frac{k_3}{\Omega} \epsilon, \quad x_j \in \mathbf{x} \backslash \mathbf{x}^c, \\
	&\operatorname{diam}\left(E_{\epsilon, \Omega}^{+,\mathbf{a}, j}(F)\right)= \operatorname{diam}\left(E_{\epsilon, \Omega T}^{+,\mathbf{a}, j}(SC_T(F))\right) \geqslant k_4 \epsilon, \quad x_j \in \mathbf{x} \backslash \mathbf{x}^c.
\end{aligned}
\]

This finishes the proof of Theorem $2.7$ with $C_1^{\prime}=\max \left(\frac{k_1}{(p-1)^{2 p-2}}, k_3\right)$, $C_2^{\prime}=\max \left(k_4, \frac{k_2}{(p-1)^{2 p-1}}\right), C_3^{\prime}=c_1(p-1)^{2 p-1}, C_4^{\prime}=\frac{1}{d}$ and $C_5^{\prime}=2$.

\subsection{Proof of Proposition \ref{prop:diameterlowerestimate1}}\label{section:proofproplower}
We separate the proof into three steps. \\
\textbf{Step 1.}\\
In this step we prove the following theorem. 
\begin{thm}\label{thm:lowerexampleconstru1}
Given the parameters $0<h \leqslant 2,0<\tau \leqslant 1,0<m \leqslant M<\infty$, let the signal $F=(\mathbf{a}, \mathbf{x}) \in \mathcal{P}_p^+$ with $\mathbf{a} \in (\mathbb{R}^+)^p$ form a single uniform cluster as follows:
\begin{itemize}
\item (centered) $x_p=-x_1$;
\item (uniform) for $1 \leqslant j<k \leqslant p$ we have
$$
\tau h \leqslant\left|x_j-x_k\right| \leqslant h ;
$$
\item $m \leqslant\left\|a_j\right\| \leqslant M$.
\end{itemize}
Then, there exist constants $K_1, \ldots, K_5$ depending only on $(d, \tau, m, M)$ such that for every $\epsilon<K_5 h^{2 d-1}$, there exists a signal $F_\epsilon=(\mathbf{b}, \mathbf{y}) \in \mathcal{P}_p^+$ satisfying the following conditions:
\begin{enumerate}
\item[(i)] $m_k(F)=m_k\left(F_\epsilon\right)$ for $k=0,1, \ldots, 2p-2$, where 
\begin{equation}\label{equ:defiofmk}
	m_k(F):=\sum_{j=1}^p a_j x_j^k;
	\end{equation}
\item[(ii)] $m_{2p-1}\left(F_\epsilon\right)=m_{2p-1}(F)+\epsilon$;
\item[(iii)] $K_1 h^{-2p+2} \epsilon \leqslant\|\mathbf{x}-\mathbf{y}\| \leqslant K_2 h^{-2p+2} \epsilon$;
\item[(iv)] $K_3 h^{-2p+1} \epsilon \leqslant\|\mathbf{a} - \mathbf{b}\| \leqslant K_4 h^{-2p+1} \epsilon$.
\end{enumerate}
\end{thm}

\begin{proof}
The case for $F=(\vect a, \vect x)\in \mathcal P_{p}$ and $F_{\epsilon} =(\vect b, \vect y)\in \mathcal P_p$ is the Theorem 6.2 in \cite{batenkov2019super}.  Now we prove that for the case when $\vect a \in (\mathbb{R}^+)^p$, from condition (i) in the theorem we actually have $\vect b \in (\mathbb{R}^+)^p$.  This is enough to prove the theorem. 

Let $b_j$'s be elements in $\vect b$ and $y_j$'s be elements in $\vect Y$. We first consider the case when $x_{j^*} = y_{q^*}$ for certain $j^*,q^*$. If $b_{q^*} \neq a_{j^*}$, then condition (i) in the theorem yields that 
\begin{equation}\label{equ:prooflowereq0}
B\beta = A\alpha, 
\end{equation}
where $\alpha =(a_1, \cdots,\ a_{j^*-1}, \ a_{j^*}-b_{q^*},\ a_{j^*+1}, \cdots, a_{p})^{\top}$, $\beta = (b_1, \cdots, b_{q^*-1}, b_{q^*+1}, \cdots, b_{p})$ and 
\[
A=\big(\phi_{2p-2}(x_1),\cdots,\phi_{2p-2}(x_{p})\big), \ B=\big(\phi_{2p-2}(y_1),\cdots,\phi_{2p-2}(y_{q^*-1}), \phi_{2p-2}(y_{q^*+1}),\cdots,  \phi_{2p-2}(y_{p})\big),
\]
with $\phi_{2p-2}(\cdot)$ being defined by (\ref{equ:phiformula}). Since all the elements in $\alpha$ are nonzero by $b_{q^*} \neq a_{j^*}$ and $a_j>0, 1\leq j\leq p$, $A$ contains $p$ different Vandermonde vectors ($\phi_{2p-2}(\cdot)$), and $B$ contains at most $p-1$ Vandermonde vectors, it is impossible to have (\ref{equ:prooflowereq0}) by \cite[Theorem 3.12]{liu2021mathematicaloned}. Thus, we must have $b_{q^*}= a_{j^*}>0$ for $x_{j^*} =  y_{q^*}$. Next, we prove that $b_{q}>0$ for those $y_q$ that are not equal to any of the $x_j$'s. Without loss of generality, we can actually assume that all the $y_q$'s are not equal to any of the $x_j$'s.  Further, since $(\vect b, \vect y)\in \mathcal P_{p}$ and $(\vect a, \vect x)\in \mathcal P_{p}^+$, all the $y_q$'s and $x_{j}$'s are distinct from each other.  We now claim that 
\begin{equation}\label{equ:prooflowereq3}
	x_1< y_1< x_2< y_2<\cdots< x_p< y_p, \  \text{ or }\ 	y_1< x_1< y_2< x_2<\cdots< y_p< x_p.
\end{equation}
We denote the $x_j, y_j$'s from left to right by $t_1<t_2<\cdots < t_{2p}$ and the corresponding $a_j , -b_j$'s by $\alpha_1, \cdots, \alpha_{2p}$.  By condition (i), it follows that 
\[
A\alpha=0, 
\]
where $\alpha =(\alpha_1, \cdots, \alpha_{2p})^{\top}$ and $A=\big(\phi_{2p-2}(t_1),\cdots,\phi_{2p-2}(t_{2p})\big)$ with $\phi_{2p-2}(\cdot)$ being defined by (\ref{equ:phiformula}). Thus we can have 
\[
-\alpha_{2p}\phi_{2p-2}(t_{2p}) = \big(\phi_{2p-2}(t_{1}), \cdots, \phi_{2p-2}(t_{2p-1})\big)(\alpha_{1},\cdots, \alpha_{2p-1})^{\top},
\]
and hence
\begin{equation}\label{equ:prooflowereq4}
	-\alpha_{2p} \left(\phi_{2p-2}(t_{1}), \cdots, \phi_{2p-2}(t_{2p-1})\right)^{-1}\phi_{2p-2}(t_{2p}) = (\alpha_{1},\cdots, \alpha_{2p-1})^{\top}.
\end{equation}
If claim (\ref{equ:prooflowereq3}) does not hold, we have $t_{q}=x_{j_{q}}$ and $t_{q+1}= x_{j_{q}+1}$ for certain $q, j_{q}$. Applying Lemma \ref{lem:invervandermonde} to (\ref{equ:prooflowereq4}) and considering the $q$-th and $(q+1)$-th elements in the vectors, we have 
\begin{align}
	&-\alpha_{2p}\Pi_{1\leq j\leq 2p-1, j\neq q}\frac{t_{2p}-t_{j}}{t_{q}-t_{j}}= \alpha_{q}, \label{equ:prooflowereq6} \\
	&-\alpha_{2p}\Pi_{1\leq j\leq 2p-1, j\neq q+1}\frac{t_{2p}-t_{j}}{t_{q+1}-t_{j}}= \alpha_{q+1}. \label{equ:prooflowereq7}
\end{align}
Observe first that, for $1\leq k\leq 2p-1$, $\Pi_{1\leq j\leq 2p-1, j\neq k}(t_{2p}-t_j)$ is always positive. Moreover, it is obvious that $\Pi_{1\leq j\leq 2p-1, j\neq q}(t_{q}-t_j)$ and $\Pi_{1\leq j\leq 2p-1, j\neq q+1}(t_{q+1}-t_j)$ have different signs in (\ref{equ:prooflowereq6}) and (\ref{equ:prooflowereq7}), respectively. It follows that $\alpha_{q}$ and $\alpha_{q+1}$ are of different signs. But the $\alpha_q$ and $\alpha_{q+1}$ are amplitudes of positive sources located at respectively $x_q$ and $x_{q+1}$, which yields a contradiction. Thus the case that $t_{q}= x_{j_{q}}$ and $t_{q+1}= x_{j_{q}+1}$ for certain $q, j_{q}$ will not happen and the claim (\ref{equ:prooflowereq3}) is thus proved.

Suppose 
\begin{equation}\label{equ:prooflowereq4.5}
y_1< x_1< y_2< x_2<\cdots< y_p< x_p, 
\end{equation}
we now prove that the $b_j$'s are all positive. The another case can be demonstrated in the same way as below. By this setting, 
\[
t_{2j-1} = y_j, \quad t_{2j} = x_j, \quad  j=1,\cdots, p.  
\]
Since for $j =1, \cdots, 2p-1$, we have 
\begin{align}\label{equ:prooflowereq5}
-\alpha_{2p}\Pi_{1\leq k\leq 2p-1, k\neq j}\frac{t_{2p}-t_{k}}{t_{j}-t_{k}}= \alpha_{j}. 
\end{align}
For $1\leq j\leq 2p-1$, $\Pi_{1\leq k\leq 2p-1, k\neq j}(t_{2p}-t_k)$ is always positive. For $j=2p-1$, since $\alpha_{2p}=a_p>0$, $-\alpha_{2p} \Pi_{1\leq k\leq 2p-1, k\neq 2p-1}(t_{2p-1}-t_k)$ is negative in (\ref{equ:prooflowereq5}). Thus we have $\alpha_{2p-1}<0$. In the same fashion, we see that $\alpha_{j}>0$ for even $j$ and $\alpha_j<0$ for odd $j$. Note that by the setting (\ref{equ:prooflowereq4.5}), $\alpha_{2j-1} = -b_j, \ j=1, \cdots, p$. This proves that $F_{\epsilon}$ is actually a positive signal and completes the proof of Theorem \ref{thm:lowerexampleconstru1}. 
\end{proof}

\textbf{Step 2.} Now we start to prove Proposition \ref{prop:diameterlowerestimate1}. It is similar to the proof in \cite{batenkov2019super}. Define $F^c$ and $F^{n c}$ to be the cluster and the non-cluster parts of $F\in \mathcal P_{d}^+$, respectively, i.e.,
\[
\begin{aligned}
	F^c &=\sum_{x_j \in \mathbf{x}^c} a_j \delta\left(x-x_j\right), \\
	F^{n c} &=\sum_{x_j \in \mathbf{x} \backslash \mathbf{x}^c} a_j \delta\left(x-x_j\right) .
\end{aligned}
\]

We first analyze the non-cluster nodes.  We construct that 
\[
F_{\epsilon}^{nc} = \sum_{x_j \in \mathbf{x} \backslash \mathbf{x}^c} a_j' \delta\left(x-x_j'\right)
\]
where $a_j^{\prime}=a_j+\frac{\epsilon}{4(d-p)}$ and $x_j^{\prime}=x_j+\frac{\epsilon}{8 \pi \Omega M (d-p)}$. For $|s| \leqslant \Omega$, the difference between the Fourier transforms of $F_{\epsilon}^{nc}$ and $F^{nc}$ satisfies
\begin{equation}\label{equ:nonclusterapprox}
\begin{aligned}
	\left|\mathcal{F}[F_{\epsilon}^{nc}](s)-\mathcal{F}[F^{nc}](s)\right| &\leq \sum_{x_j \in \mathbf{x} \backslash \mathbf{x}^c}\left|a_j e^{-2 \pi i x_j s}-a_j^{\prime} e^{-2 \pi i x_j^{\prime} s}\right| \\
	& \leqslant \sum_{x_j \in \mathbf{x} \backslash \mathbf{x}^c}\left(\left|a_j e^{-2 \pi i x_j s}\left(1-e^{-2 \pi i \frac{\epsilon}{8 \pi \Omega M(d-p)} s}\right)\right|+ \frac{\epsilon}{4(d-p)}\right) \\
	& \leqslant \frac{\epsilon}{4}+\frac{\epsilon}{4}=\frac{\epsilon}{2}.
\end{aligned}
\end{equation}
Note that $a_j'>0$ since $a_j^{\prime}=a_j+\frac{\epsilon}{4(d-p)}$ and $F_{\epsilon}^{nc}$ is a positive signal. 

We next analyze the cluster nodes. We suppose that $x_1+x_p=0$. For the case when $x_1+x_p\neq 0$, utilizing decomposition
\[
\sum_{j=1}^{p} a_j e^{-2\pi i x_js} = \sum_{j=1}^{p} a_j e^{-2\pi i \frac{x_1+x_p}{2} s} e^{-2\pi i \left(x_j - \frac{x_1+x_p}{2}\right) s},
\]
we nearly transfer the problem to the case when $x_1+x_p=0$ and it is not difficult to see that the following arguments hold as well with only a small modification, which is enough to prove the proposition. 

Next, we define a blowup of $F^c$ by $\Omega$ as 
\[
F_{(\Omega)}^c=S C_{\frac{1}{\Omega}}\left(F^c\right)=\sum_{x_j \in \mathrm{x}^c} a_j \delta\left(x-\Omega x_j\right) 
\]
where $SC$ is defined by Definition \ref{defi:scaling}. Let $\tilde{h}=\Omega h$ and $c_1=K_5(p, \tau, m, M)$ as in Theorem \ref{thm:lowerexampleconstru1}. Let $\epsilon \leqslant c_1(\Omega h)^{2 p-1}$. Now, we apply Theorem \ref{thm:lowerexampleconstru1} with parameters $p, \tilde{h}, \tau, m, M, \tilde{\epsilon}=c_2 \epsilon$ and the signal $F_{(\Omega)}^c$, where $c_2 \leqslant 1$ will be determined below. We can obtain a signal $F_{(\Omega), \epsilon}^c \in \mathcal P_{p}^+$ such that the following hold for the difference of signals $H=F_{(\Omega), \epsilon}^c-F_{(\Omega)}^c$ :
\begin{equation}\label{equ:prooflowereq10}
m_k(H)=0, \quad k=0,1, \ldots, 2 p-2,\quad  m_{2 p-1}(H)=\tilde{\epsilon} ;
\end{equation}
while also, for some $j_1, j_2 \in\{1, \ldots, p\}$
\begin{equation} \label{equ:prooflowereq9}
\begin{aligned}
&\left|P_{\mathbf{x}, j_1}\left(F_{(\Omega), \epsilon}^c\right)-P_{\mathbf{x}, j_1}\left(F_{(\Omega)}^c\right)\right| \geqslant K_1(\Omega h)^{-2 p+2} \tilde{\epsilon},\\
&\left|P_{\mathbf{x}, j}\left(F_{(\Omega), \epsilon}^c\right)-P_{\mathbf{x}, j}\left(F_{(\Omega)}^c\right)\right| \leqslant K_2(\Omega h)^{-2 p+2} \tilde{\epsilon}, \quad j=1, \ldots, p,\\
&\left|P_{\mathbf{a} j_2}\left(F_{(\Omega), \epsilon}^c\right)-P_{\mathbf{a}, j_2}\left(F_{(\Omega)}^c\right)\right| \geqslant K_3(\Omega h)^{-2 p+1} \tilde{\epsilon}.
\end{aligned}
\end{equation}

Now, considering 
\[
F_{\epsilon}^c=S C_{\Omega}\left(F_{(\Omega), \epsilon}^c\right) , 
\]
we obtain that
\begin{align*}
&\left|P_{\mathbf{x}, j_1}\left(F_{\epsilon}^c\right)-P_{\mathbf{x}, j_1}\left(F^c\right)\right| \geqslant \frac{K_1}{\Omega}(\Omega h)^{-2 p+2} \tilde{\epsilon},\\
&\left|P_{\mathbf{a}, j_2}\left(F_{\epsilon}^c\right)-P_{\mathbf{a}, j_2}\left(F^c\right)\right| \geqslant K_3(\Omega h)^{-2 p+1} \tilde{\epsilon}. 
\end{align*}
From the above definitions, we have $H_{\Omega}=S C_{\Omega}(H)=F_{\epsilon}^c-F^c$. We next show that there is a choice of $c_2$ such that
\begin{equation}\label{equ:prooflowereq11}
\left|\mathcal{F}[H_{\Omega}](s)\right| \leqslant \frac{\epsilon}{2}, \quad|s| \leqslant \Omega .
\end{equation}
Put $\omega=s / \Omega$. Then, by $\mathcal{F}\left[H_{\Omega}\right](s)=\mathcal{F}[H](\omega)$, the above inequality is equivalent to 
\begin{equation}\label{equ:prooflowereq8}
|\mathcal{F}[H](\omega)|\leq \frac{\epsilon}{2}, \quad |\omega|\leq 1.  
\end{equation}

\textbf{Step 3}. In this step we prove that (\ref{equ:prooflowereq8}) holds for a choice of $c_2$. Note that we have the following Taylor expansion of $\mathcal F[H](\omega)$:
\begin{equation}\label{equ:taylorexpan1}
\mathcal{F}[H](\omega)=\sum_{k=0}^{\infty} \frac{1}{k !} m_k(H)(-2 \pi i \omega)^k .
\end{equation}
Next we apply the following Taylor domination property \cite[Theorem 6.3]{batenkov2019super}.

\begin{thm}\label{thm:taylordominant1}
Let $H=\sum_{j=1}^{2p} \beta_j \delta\left(x-t_j\right)$, and put $R=\min _{j=1, \ldots, 2 p}\left|t_j\right|^{-1}>0$. Then, for all $k \geqslant 2 p$, we have the so-called Taylor domination property
\[
\left|m_k(H)\right| R^k \leqslant\left(\frac{2 e k}{2 p}\right)^{2 p} \max _{\ell=0,1, \ldots, 2 p-1}\left|m_{\ell}(H)\right| R^{\ell} .
\]
\end{thm}

Recall that $H=F_{(\Omega), \epsilon}^c-F_{(\Omega)}^c$. By Definition \ref{defi:clusterconfig}, the nodes of $F_{(\Omega)}^c$ is inside the interval $\left[-\frac{\Omega h}{2}, \frac{\Omega h}{2}\right]$. The nodes of $F_{(\Omega), \epsilon}^c$, by (\ref{equ:prooflowereq9}), satisfy
$$
\begin{aligned}
	\left|P_{\mathbf{x}, j}\left(F_{(\Omega), \epsilon}^c\right)\right| \leqslant \frac{\Omega h}{2}+K_2(\Omega h)^{-2 p+2} \tilde{\epsilon} \leqslant \frac{\Omega h}{2}+K_2(\Omega h)^{-2 p+2} c_1(\Omega h)^{2 p-1} =(\Omega h)\left(c_1 K_2+\frac{1}{2}\right) .
\end{aligned}
$$
Since $\Omega h \leqslant 2$ by assumption, we can conclude that the factor $R$ in Theorem \ref{thm:taylordominant1} is greater than $C_4=\frac{1}{2\left(c_1 K_2+\frac{1}{2}\right)}$. 

Now, we continue the proof of (\ref{equ:prooflowereq8}). By Theorem \ref{thm:taylordominant1} and (\ref{equ:prooflowereq10}), we have for $k \geqslant 2p$,
\begin{align*}
	\left|m_k(H)\right|  \leqslant\left(\frac{e}{p}\right)^{2 p} k^{2 p} R^{2 p-1-k} \tilde{\epsilon} \leqslant C_5 C_4^{2 p-1-k} k^{2 p} \tilde{\epsilon}. 
\end{align*}
Plugging this into (\ref{equ:taylorexpan1}) we obtain that
$$
|\mathcal{F}(H)(\omega)| \leqslant \frac{\tilde{\epsilon}|2 \pi \omega|^{2 p-1}}{(2 p-1) !}+C_5 C_4^{2 p-1} \tilde{\epsilon} \sum_{k \geqslant 2 p}\left(\frac{2 \pi|\omega|}{C_4}\right)^k \frac{k^{2 p}}{k !}.
$$
Let $\zeta=\frac{2 \pi|\omega|}{C_4}$ and by $|\omega|\leq 1$, we further have 
\begin{align*}
	|\mathcal{F}(H)(\omega)|  \leqslant C_6 \tilde{\epsilon} \sum_{k \geqslant 2 p-1} \zeta^k \frac{k^{2 p}}{k !} \leqslant C_7 \tilde{\epsilon} .
\end{align*}
Therefore, we can choose $c_2=\min \left(1, \frac{1}{2C_7}\right)$ to ensure that
$$
|\mathcal{F}(H)(\omega)| \leqslant \frac{\epsilon}{2}, \quad|\omega| \leqslant 1,
$$
which shows (\ref{equ:prooflowereq11}). 

Finally, we construct the positive signal $F_\epsilon=F_{\epsilon}^{nc}+F_{\epsilon}^c$. Thus we have 
\[
\left|F_\epsilon(s) -F(s) \right|\leq  \left|F_\epsilon^{nc}(s) -F^{nc}(s) \right| + \left|F_\epsilon^{c}(s) -F^{c}(s) \right|\leq \epsilon, \  s\in [-\Omega, \Omega]. 
\]
This completes the proof of Proposition \ref{prop:diameterlowerestimate1} with $k_1=K_1, k_2=K_3, k_3 = \frac{1}{8\pi\Omega M(d-p)}, k_4 = \frac{1}{4(d-p)}$.
\end{proof}

\bibliographystyle{plain}
\bibliography{reference_final} 
\end{document}